\documentclass{article}
\usepackage[margin=1.5in]{geometry}
\usepackage{amsmath,amsthm,amssymb}
\usepackage{bbm}
\usepackage{physics}
\usepackage{color}
\usepackage{algorithm}
\usepackage{algorithmicx}
\usepackage{algpseudocode}
\usepackage{xspace}

\usepackage{comment}

\title{Submodular Stochastic Probing with Prices}

\newcommand{\greedyProb}{\textsf{GreedyProbing}\xspace}
\newcommand{\price}{\Delta}
\newcommand{\z}{\vb{z}}
\renewcommand{\v}{\vb{v}}
\newcommand{\wC}{C}
\newcommand{\cost}{\textsf{cost}}

\newcommand{\p}{\vb{p}}
\newcommand{\e}{\vb{e}}
\newcommand{\eps}{\epsilon}

\newcommand{\R}{\mathbb{R}}
\newcommand{\alg}{\textsf{Alg}}
\newcommand{\opt}{\textsf{Opt}}
\newcommand{\U}{\mathcal{N}} % Universe
\newcommand{\N}{\U}          % Universe

\newcommand{\I}{\mathcal{I}}
\newcommand{\Iin}{\I_{\textnormal{in}}}
\newcommand{\Iout}{\I_{\textnormal{out}}}
\newcommand{\bpi}{\boldsymbol{\price}}
\renewcommand{\P}{\mathcal{P}}
\newcommand{\ind}{\textbf{1}}
\newcommand{\supp}{\textnormal{supp}}
\newcommand{\x}{\vb{x}}
\newcommand{\y}{\vb{y}}
\renewcommand{\i}{\textnormal{in}}
\renewcommand{\o}{\textnormal{out}}

\newcommand{\spp}{\textsf{SPP}}
\renewcommand{\sp}{\textsf{SP}}

\newcommand{\piout}{\pi^{\o}}
\newcommand{\piin}{\pi^{\i}}
\newcommand{\w}{\vb{w}}

\newcommand{\act}{\mathcal{A}}
\newcommand{\EE}{\mathcal{E}}
\newcommand{\joint}{J}

\DeclareMathOperator{\Prob}{\mathbb{P}}
\renewcommand{\Pr}{\Prob}

\newtheorem{lemma}{Lemma}
\newtheorem{theorem}{Theorem}

\newtheorem{corollary}{Corollary}
\theoremstyle{definition}
\newtheorem{definition}{Definition}
\theoremstyle{definition}
\newtheorem{observation}{Observation}

\DeclareMathOperator{\argmax}{\textnormal{argmax}}
\DeclareMathOperator{\E}{\textnormal{\bf E}}

\title{Submodular Stochastic Probing with Prices}
\author{Ben Chugg\footnote{benjamin.chugg@maths.ox.ac.uk. Supported by a research internship at RIKEN.}\qquad Takanori Maehara\footnote{takanori.maehara@riken.jp}}
\date{}

\begin{document}

\maketitle

\begin{abstract}
We introduce \emph{Stochastic Probing with Prices (SPP)}, a variant of the \emph{Stochastic Probing (SP)} model in which we must pay a price to probe an element.
A SPP problem involves two set systems $(\N,\Iin)$ and $(\N,\Iout)$ where each $e\in\N$ is \emph{active} with probability $p_e$.
To discover whether an element $e$ is active, it must be \emph{probed} by paying the \emph{price} $\price_e$.
If an element is probed and is active, then it is irrevocably added to the solution. 
Moreover, at all times, the set of probed elements must lie in $\Iout$, and the solution (the set of probed and active elements) must lie in $\Iin$.
The goal is to maximize a submodular set function $f$ minus the cost of the probes. 
We give a bi-criteria approximation algorithm to the online version of this problem, in which the elements are shown to the algorithm in a possibly adversarial order. 
% If $\Iin$ and $\Iout$ admit $(b,c_\i)$ and $(b,c_\o)$ contention resolution schemes, we give a $\max\{c_\o c_\i,c_\o+c_\i-1\}\alpha(b)$-approximation to online SPP, where $\alpha(b)=1-e^{-b}$ if $f$ is monotone, and $be^{-b}$ otherwise.
% These results apply in the online version of the problem: The elements are presented in an arbitrary (and potentially adversarial) order.
% In the special case when $\Iin$ and $\Iout$ are the intersection of $k$ and $\ell$ matroids, we demonstrate that the optimal value for $b$ is $\frac{1}{2}(z+2-\sqrt{z(z+4)})$ when $f$ is monotone, and $z+1-W(ze^{z+1})$ when $f$ is non-monotone, where $z=k+\ell$ and $W$ is the lambert W function. 
Our results translate to state-of-the-art approximations for the traditional (online) stochastic probing problem. 
%when $f$ is monotone, and provide the first approximation to online, adversarial SP when $f$ is non-monotone. 
%Finally, we demonstrate that when $f$ is modular there exists an $\alpha$-approximation to SP only if there exists an $\alpha$-approximation to SPP. \\
\end{abstract}

%%%%%%%%%%%%%%%%%%%%%%%%%%%%%%%%%%%%%%%%%%%%%%%%%%%%%%%%%%%%%%%%%%%%%%%%%%%%%%%%
\section{Introduction}

Life is hard: There is lots of uncertainty, it is rarely the case that one can gather all the information before making a decision, and information isn't free. 
Whether we're buying a plane ticket for a conference, or a pair of socks for a significant other\footnote{We leave the question of whether a pair of socks is the optimal gift for another paper.}, it is impossible to be aware of every item in the market and gathering the information to make a good decision costs time and resources. 
The job of computer scientists is to try and model such scenarios.
In this paper, we propose the \emph{Stochastic Probing with Prices (SPP)} model, which combines two previous models---Stochastic Probing~\cite{gupta2013} and the Price of Information~\cite{singla2018}---in order to study decision making under uncertainty where ascertaining new information has a cost.

\subsection{Problem Formulation}

Let $\U=\{e_1,\dots,e_n\}$ be a set of $n$ elements, and $(\N,\Iin)$, $(\N,\Iout)$ be two downward-closed set systems\footnote{$\I \subset 2^\N$ is \emph{downward-closed} if $A\in\I$ and $B\subset A$ implies $B\in \I$.}. Every element $e\in\U$ is \emph{active} with probability $p_e$, independently of all other elements.  
In order to determine whether $e$ is active, we can \emph{probe} $e$. However, if $e$ is probed and is active, then it must be added to our current solution (which thus consists of all previously probed elements which are also active). Moreover, at all times the set of probed elements must lie in $\Iout$ and the solution must lie in $\Iin$. 
Thus, if $P \subset \N$ is the set probed elements so far and $S\subset P$ is the solution, we can only query an element $e$ that satisfies $P\cup\{e\}\in\Iout$ and $S\cup\{e\}\in\Iin$. 
The set of active elements is determined a priori, i.e., before the execution of the algorithm. Given an instance of the problem, let $\act$ refer to the set of active elements. 
We suppose we are given two set-functions, $f,\cost:2^\U\to\R$, where $f$ is a \emph{utility} or \emph{objective} function, and $\cost$ is (naturally) a \emph{cost} or \emph{payment} function.  
The goal is to query the elements in such a way as to maximize $f$ while minimizing the cost paid for the queries, i.e., to maximize 
$\E_{\act}[f(P\cap\act)-\cost(P)]$. 
In this paper we assume that the cost function is linear: for each $e\in\N$ there exists a price $\price_e\in\R$ drawn from some distribution $F_e$, and $\cost(P)=\sum_{e\in P}\price_e$. 
For simplicity, we will assume that $\price_e$ is deterministic (i.e., $F_e$ is a point mass).
Otherwise, since $\cost$ is linear, we can simply replace $\price_e$ with $\E[\price_e]$ to obtain our results for more general distributions. We call this model \emph{Stochastic Probing with Prices (SPP)}. 
If we can choose the order in which the data is revealed to us is, the problem is \emph{offline}, otherwise it is \emph{online}. 
We draw a distinction between two kinds of online settings: an \emph{adversarial setting}, in which the order may be chosen by an adversary, and a \emph{random-order setting}, in which the elements are chosen uniformly at random.  In this paper, all of our results pertaining to online SPP apply in the adversarial online setting. 
Surprisingly, when $f$ is submodular, our algorithms do not seem to yield better approximations in the offline or  random-order settings than in the adversarial setting.

SPP is a generalization of the \emph{Stochastic Probing (SP)} model of Gupta and Nagarajan~\cite{gupta2013} and borrows the pricing idea from the recently developed \emph{Price of Information (PoI)} model of Singla~\cite{singla2018}---see Sections \ref{sec:intro_sp} and \ref{sec:intro_poi}.
It falls under the general category of what we might call a \emph{query-and-commit} model of computation, in which irrevocable decisions must be made on the fly. 
Stochastic matching, packing and the secretary problem are examples of problems studied under such a model \cite{feldman2009online,karp1990optimal,dean2005adaptivity,chow1964optimal}.
Since SPP is a generalization of SP (as we will see), it is useful in the same application domains as the latter model.
However, we believe few problems encountered in practice do not have a cost associated with obtaining information (even if it is only a time cost).  
%\ben{Can we find more applications of SPP?}
%\tm{The online dating (see gupta2013 \url{http://citeseerx.ist.psu.edu/viewdoc/download?doi=10.1.1.369.8795&rep=rep1&type=pdf}) will be a concrete example. }
%\tm{Also, I think online advertising \url{https://arxiv.org/pdf/0905.4100.pdf} will be an example: it is a bipartite matching problem where the left vertices are advertisements and the right vertices are users. Putting advertisements is not free.}

%\tm{a comment by my friend: if the cost is a hard constraint, the problem is represented by SP. Thus, it would be nice to add motivation to consider this setting. (not necessarily)}
%\ben{hard constraint?}
%\tm{i.e., why just adding constraint to cost(Q) \le const. is not suffice?}
We now formally introduce the models upon which SPP is built, and relevant known results. 

\subsection{Stochastic Probing}
\label{sec:intro_sp}
The SP model is precisely the SPP model without the cost function. That is, the goal is to maximize $f$ only. It has been a valuable abstraction of stochastic matching and packing, and has been useful in modeling such practical problems as kidney exchange, online dating, and posted price mechanisms \cite{gupta2013,adamczyk2016,adamczyk2015}. 

In order the present many of the known results for SP, we need to introduce \emph{contention resolution schemes (CRSs)}, first introduced by Chekuri et al.~\cite{chekuri2011}. Suppose we are given a fractional solution $\x$ to an LP-relaxation of some constrained maximization problem. The goal is now to round $\x$ to an integral solution $\x_I$ which respects the constraints, without losing too much of the value given by the fractional solution. Intuitively, a $c$-CRS gives a guarantee that if $x_e>0$ then $x_e$ will be in the solution with probability at least $cx_e$. Thus, if $L$ is the linear (or concave) objective function, then $L(\x_I)\geq cL(\x)$, with high probability (w.h.p.).  A $(b,c)$-CRS gives the same guarantee assuming that $x\in b\P$, where $\P$ is a polytope relaxation of the constraints of the problem. (Thus, a $(1,c)$-CRS is a $c$-CRS.) Feldman et al.~\cite{feldman2016online} recently extended the notion of CRSs to online settings. We will provide the formal definition of both offline and online CRSs in Section \ref{sec:prelims_crschemes}.

As mentioned above, the general framework of offline SP problems was first introduced by \cite{gupta2013} who restrict their attention to modular (linear) objective functions. When $\Iin$ and $\Iout$ admit offline $(b,c_\i)$ and $(b,c_\o)$ schemes, they give a $b(c_\o+c_\i-1)$-approximation to modular SP. This implies a $\Omega(\frac{1}{(k+\ell)^2})$-approximation when $\Iin$ and $\Iout$ are $k$ and $\ell$ systems \cite{gupta2013}. 
If $\Iin$ and $\Iout$ are the intersection of $k$ and $\ell$ matroids, respectively, then they give a $\frac{1}{4(k+\ell)}$-approximation. Under these same constraints, Adamczyk et al.~\cite{adamczyk2016} improve this approximation to $\frac{1}{k+\ell}$ via a randomized rounding technique, and give a $\frac{1-1/e}{k+\ell+1}$-approximation in the more general case when $f$ is monotone submodular.
Online SP was first introduced and studied by Feldman et al. \cite{feldman2016online}. When $\Iin$ and $\Iout$ admit $(b,c_\i)$ and $(b,c_\o)$ online CRSs respectively, they give a $b c_\o c_\i$ approximation to online modular SP, and a $(1-e^{-b})c_\i c_\o$-approximation to online, submodular SP for monotone objective functions. These results hold in the adversarial setting. 
More recently, Adamczyk and W\l{}odarcyzk~\cite{adamczyk2018random} initiated the study of random order CRSs, and use such schemes to achieve the first approximations to non-monotone submodular SP in the random-order setting, giving a $\frac{1}{(k+\ell+1)e}$-approximation.

\subsection{Price of Information}
\label{sec:intro_poi}

%Incorporating prices---or a cost function more generally---has seen less attention as a model than query-and-commit. 
Perhaps the first, well-known problem to incorporate costs for querying elements is Weitzman's Pandora's Box \cite{weitzman1979optimal}. 
Generalizing this problem,  Singla~\cite{singla2018} recently developed 
the PoI model.

%\tm{``however'' appears twice. I think after ``However'' in the above paragraph is merged below. }
%Recently however, Singla~\cite{singla2018} developed an elegant model named the \emph{Price of Information (PoI)}. 
%We use the PoI model as inspiration in this paper. 
In this setting, we are given a single set system\footnote{Although it is possible to generalize the model, as discussed in \cite{singla2018}.} $(\U,\I)$ and each element $e$ has a value $v_e$ which is drawn from some distribution $D_e$. The distribution is not limited to a Bernoulli distribution as in the SP model. Each element $e$ also a \emph{price}, $\price_e$. 
The goal is to query a set of elements $P$ in order to maximize $\E_{\{D_e\}}[\max_{I\subset P,I\in\I}\{\sum_{e\in I}v_e + h(I)\} - \sum_{e\in P}\price_e]$,
where $h$ is some function that depends only on the chosen set, not the values of the elements.
The function $\sum_{e\in I}v_e+h(I)$ is called a \emph{semiadditive} function for this reason. 
Besides the fact that one must pay for making queries in this model, the key distinction is that one gets to choose the maximizing argument $I\subset P$. That is, this model is not a query-and-commit model. 
The major result of Singla~\cite{singla2018} was to prove that if there is an $\alpha$-approximation to a given problem using what he calls a ``frugal'' algorithm (essentially, but slightly more general than, a greedy algorithm) when all prices are zero, then there is an    $\alpha$-approximation to the problem with arbitrary prices. Intriguingly, in our work we will obtain a similar result for modular SPP---this will be described in detail later. 

While SPP is a generalization of SP, there is no strict inclusion relation between SPP and the PoI model. 
The PoI model is not a query-and-commit model and as such it does not capture SPP.
Furthermore, it is unclear how one would expand a PoI problem to use more general classes of objective functions---submodular functions, for example---while this is easily modeled in SPP.
Conversely, SPP does not capture PoI because it is more limited in its assumption that an element is either active or inactive, and cannot model more complicated distributions. 

In his Ph.D. thesis, Singla does examine the PoI setting with commitment constraints~\cite{singlaPHD}.
In this scenario, the values of the elements change in time as a Markov process, and at each step, we can decide to advance the Markov chain and pay a cost, or stop it and obtain its current value. 
The setting of activation probabilities can be viewed as a two-stage Markov chain; therefore this model can be viewed as a generalization in that sense.
However, there is no outer-constraint, i.e., every element can be queried.

\subsection{Our Results}
\label{sec:results}

In this paper, we are concerned with two classes of objective functions: modular and submodular. We call the corresponding problems modular and submodular SPP. 
Note that if $f$ is modular then there exists weights $(w_e)_{e\in\N}$ such that $f(E)=\sum_{e\in E}w_e$. Consequently, when all prices are zero, this problem has also been called weighted stochastic probing. 

We begin by studying modular SPP. We demonstrate that if there exists an $\alpha$-approximation for SP, then there exists an $\alpha$-approximation for SPP.
This result applies in all SPP settings (e.g., adversarial, random order, offline). 
\begin{theorem}
\label{thm:sp=spp}
If there exists an $\alpha$-approximation to modular SP then there exists an $\alpha$-approximation to modular SPP. 
\end{theorem}

For submodular SPP, the problem becomes significantly more difficult. It is known that submodular maximum facility location, which is a special case of SPP, is inapproximable to within any non-trivial factor~\cite{feige2013pass}. It is therefore highly unlikely that there exists an approximation of SPP, even one that holds with high probability. Thus, we relax our requirements and search for constant-factor bi-criteria approximations: 

\begin{definition}
\label{def:bi-criteria}
Let $\opt$ be an optimal query strategy. A query strategy $\alg$ is an $(\alpha,\beta)$-approximation to SPP if $\E[\alg]\geq \alpha \E[f(\opt\cap \act)]-\beta\cost[\opt]$, where the expectation is taken over $\act$ and any random choices of the algorithm.
\end{definition}

We demonstrate that slight variations of the algorithm for modular SP of Gupta and Naragajan~\cite{gupta2013} (using offline CRSs)  and the algorithm of Feldman et al. \cite{feldman2016online} (using online CRSs) give good bi-criteria approximations to
adversarial submodular SPP if $\Iin$ and $\Iout$ admit $(b,c_\i)$ and $(b,c_\o)$ CRSs. Interestingly, our variation of Gupta and Naragajan's algorithm increases its scope in two significant ways: we do not require the inner CRS to be ordered, nor the ability to choose the order in which the elements are presented. Thus, our version of their algorithm (which is called \textbf{Offline-Rounding} because it uses offline CRSs---Section \ref{sec:policy}) applies to online SPP.  

\begin{theorem}
\label{thm:spp}
Suppose $\P(\Iin)$ and $\P(\Iout)$ admit $(b,c_\i)$ and $(b,c_\o)$ CRSs
%\piin$ and $\piout$, 
and let $\gamma=\max\{c_\o+c_\i-1,c_\o c_\i\}$. There exists a $(\gamma\cdot\alpha(b),b)$-approximation to adversarial, submodular SPP 
%and If $\piin$ is offline, then there exists a $\gamma\cdot\alpha(b)$-approximation to offline, submodular SPP, 
where $\alpha(b)=1-e^{-b}$ if the objective function is monotone, and $be^{-b}$ otherwise. 
\end{theorem}

A note now on how our results translate to SP. When all prices are zero, Theorem~\ref{thm:spp} gives a $\gamma\cdot\alpha(b)$-approximation for online, submodular SP. In this case, if $\Iin$ and $\Iout$ are the intersection of $k$ and $\ell$ matroids, we demonstrate that the optimal value of $b$ is $\frac{1}{2}(k+\ell+2-\sqrt{(k+\ell)(k+\ell+4)})$ when $f$ is non-monotone, and $k+\ell+1-W((k+\ell)e^{k+\ell+1})$ when $f$ is monotone, where $W$ is the Lambert W function \cite{corless1996lambertw}.\footnote{$W$ is defined by $x = W(x) e^{W(x)}$. It is also called the product log function.} This gives a $(W(\rho)-k-\ell)^{k+\ell+1}/W(\rho)$-approximation in the monotone case, where $\rho=(k+\ell)e^{k+\ell+1}$, and something significantly uglier in the non-monotone case. These results are given in Section \ref{sec:approximations}.
These results outperform those of Feldman et al.~\cite{feldman2016online} (who study the adversarial setting) and Adamczyk and W\l{}odarcyzk~\cite{adamczyk2018random} (who study the random-order setting). 
This is also shown in Section \ref{sec:approximations}.
In the offline case, the approximation ratio of Adamczyk et al.~\cite{adamczyk2016} beats this approximation ratio, however,  
we remark that our results are the first pertaining to a non-monotone objective function in the adversarial setting. 

We hope our results provide a unified view of SP and SPP problems. The analyses use existing techniques which should be familiar to those working in submodular maximization which we hope makes the results easily accessible. 

\subsection{Proof Techniques}
For a modular objective function, there is little daylight between SP and SPP. The former is reduced to the latter by modifying the weights.
No such transformation can be made in the more general case of a submodular objective function however, which thus increases the difficulty of the problem. 

Similar to the algorithms presented in \cite{gupta2013,adamczyk2016,adamczyk2018random,feldman2016online} for SP, our algorithms proceed by obtaining a fractional to a suitably chosen LP-relaxation of the problem, and then rounding according to both CRSs.
% \tm{overcome appear twice. The first one can be removed.}
%The first major difficulty that needs to be overcome is that while $\E[f(P\cap A)]-\cost(P)$ is still a submodular function, it may not be non-negative. Since traditional algorithms for submodular maximization require that the objective function be non-negative, this is a hurdle which needs to be overcome. 
The first major difficulty is that while $\E[f(P\cap A)]-\cost(P)$ is still a submodular function, it may not be non-negative. Since traditional algorithms for submodular maximization require that the objective function be non-negative, this is a hurdle which needs to be overcome. 
To do so we use the recent technique of Sviridenko et al.~\cite{sviridenko2017} in order to maximize a function of the form $f^+(\x\circ\p)- C(\x)$, where $f^+$ is concave and $C$ is linear.
Lemma \ref{lem:LP-relax} demonstrates that this objective function (subject to the appropriate constraints) provides an upper bound on the value achieved by the optimal policy.
%The guarantee is given by Theorem \ref{thm:submod+linear}.

\section{Preliminaries}
\label{sec:prelims}
Let $[n]=\{1,\dots,n\}$. Given a vector $\x\in[0,1]^\U$, let $R(\x)$ be a random set where each $e\in\U$ is selected with probability $x_e$. We will sometimes write $R\sim \x$ to refer to choosing a random set $R(\x)$. We call $p_e$ the \emph{activation probability of $e\in\N$}. We emphasize that it is independent of the activation probabilities of other elements.  
For a set $E\subset\N$ and element $e\in\N$, we will often write $E+e$ in place of the more cumbersome $E\cup\{e\}$. 
Given a (stochastic) probing algorithm $\alg$, we will abuse notation somewhat and write $\E[\alg]$ to mean the expected value of the solution obtained by querying according to $Q$, i.e., $\E[f(\alg\cap\act)-\cost(\alg)]$. Note that the expectation is over the joint distribution $\{p_e\}$ and any randomness in the choices made by the algorithm.

\subsection{Modular and Submodular Functions}
A function $f:2^\U\to\R$ is \emph{submodular} if for all $A,B\subset \U$ $f(A)+f(B)\geq f(A\cap B)+f(A\cup B).$
% \begin{equation}
% \label{eq:submodular}
% f(A)+f(B)\geq f(A\cap B)+f(A\cup B).
% \end{equation}
If this inequality holds with equality, then $f$ is \emph{modular}.  
Given $E\subset\U$, write $f_E(e)$ for $f(E+e)-f(E)$.
We will consider two extensions of a set function $f:2^\U\to \R$ to $[0,1]^\U$. The first is the \emph{multilinear extension} \cite{ageev2004pipage,calinescu2007maximizing}: \[F(\y)=\sum_{A\subset \U}f(A)\prod_{e\in A}y_e\prod_{e\in A^c}(1-y_e)=\E_{R\sim\y}[f(R)]\]. 
The second is the concave closure of $f$ \cite{chekuri2011}: 
\[
    f^+(\y)=\max_{S\subset\N}\bigg\{\sum_{S\subset \U}p_Sf(S):p_S\geq 0\;\forall S\subset\U,\; \sum_{S\subset\U}p_S=1,\; \sum_{S:e\in S}p_S= y_e\;\forall e\in \U\bigg\}.\] 
    
As indicated by the wording $f^+$ is concave. It is well known that $F(\y)\leq f^+(\y)$. Indeed, $f^+(\y)$ can be viewed as selecting the probability distribution over $2^\U$ maximizing $\E[f(E)]$ subject to the constraint that $\Pr[e\in E]\leq y_e$, while $F(\y)$ is the value of $\E[f(E)]$ under a particular such distribution. 

\subsection{Contention Resolution Schemes}
\label{sec:prelims_crschemes}
Let $\I\subset2^\N$ be downward-closed. The \emph{polytope relaxation of $\I$} is the set $\P(\I)\subset[0,1]^\U$ defined as the convex hull of all characteristic vectors of $\I$. 

%  Finally, given a set $E\subset\U$, we denote b $\I\restriction_{E}$ the \emph{restriction of $I$ to $E$}: $\{I\cap E:I\in\I\}$. 

For the rest of this section fix a downward closed set system $\I$, and let $\P(\I)$ be its convex relaxation. We now introduce the formal definitions of offline and online CRSs. 

\begin{definition}[Offline CRS \cite{chekuri2011}]
\label{def:CRS}
 For $b,c\in[0,1]$, a $(b,c)$ \emph{offline Contention Resolution scheme} for $\P(\I)$ is family of (possibly randomized) functions $\{\pi_{\x} :2^U\to 2^U:\x\in b\P(\I)\}$ such that for all $\x\in b\P(\I)$ and $E\subset \U$: 
(1)     $\pi_{\x}(E) \subset E\cap \supp(\x)$;
(2) $\pi_{\x}(E)\in\I$ with probability 1; and 
(3) for all $e\in\supp(\x)$, $\Pr_{R(\x),\pi}[e\in\pi_{\x}(R(\x))|e\in R(\x)]\geq c$.
\end{definition}

%Moreover, the CRS is said to be \emph{monotone} if for all $A_1\subset A_2$ and $\x\in A_1$, $\Pr[i\in \Gamma_{\x}(A_1)]\geq \Pr[i\in \Gamma_{\x}(A_2)]$.
The CRS $\pi$ is \emph{monotone} if for all $e\in E_1\subset E_2$, $\Pr_\pi[e\in \pi(E_1)]\geq \Pr_\pi[e\in \pi(E_2)]$. 
%We say a CRS is $\eps$-\emph{strict}, for $\eps\geq0$, if $c-\eps\leq \Pr[i\in \Gamma_x(E):i\in E]\leq c+\eps$. 
%Finally, the CRS is said to be \emph{ordered} if there exists an order $\sigma$ on $\U$ such that $\pi_{\x}(E)=\greedy(\sigma,\I)$.\footnote{This concept was first introduced in \cite{gupta2013}.} 
We emphasize that in condition (3) of Definition \ref{def:CRS}, the probability is over both the random set $R(\x)$ and the CRS, whereas in the definition of monotonicity of a CRS, the probability is taken only over the (possibly) random choices of scheme itself.                                   

\begin{definition}[Online CRS \cite{feldman2016online}]
For $b,c\in[0,1]$, a $(b,c)$ \emph{online contention resolution scheme} of $\P(\I)$ is a procedure which defines for any $\x\in\P(\I)$ a family $\I_{\x}\subset\I$ such that for all $e\in\U$, $\Pr[I +e\in\I_{\x} \text{ for all } I\subset R(\x), I\in\I_{\x}] \geq c$.
\end{definition}

Given an online CRS of $\P(\I)$ the authors of \cite{feldman2016online} define a related offline CRS, called the \emph{characteristic} CRS $\pi(E)=\{e\in E:I+e\in \I_{\pi,\x},\text{ for all }I\subset E,I\in\I_{\pi,x}\}$. They verify that the characteristic CRS meets the conditions of an offline CRS and demonstrate that the characteristic CRS of a $(b,c)$ online CRS is a monotone, $(b,c)$ CRS. 
If we refer to a CRS without specifying whether it is online or offline, the discussion should apply to both. 
Unless otherwise stated, we will assume that any CRS discussed in this paper is \emph{efficient}: that it can be computed in polynomial time. 
% It will also be useful for use to work with CRSs defined on subfamilies of the given inner and outer constraints. Consequently, given a CRS on $\P(\I)$ and a subset $M\subset \U$ define a CRS $\hpi\restriction_{M^c}$ on $\P(\I\restriction_{M^c})$ by $E\mapsto \pi(E\setminus M)$. We will write simply $\hpi=\hpi\restriction_{M^c}$ when $M$ and $I$ are clear from context. The next lemma verifies that $\hpi$ is indeed a CRS and obeys the same properties as $\pi$. 

% \begin{observation}
% Let $M\subset \U$. If $\pi$ is a $(b,c)$ CRS on $\P(\I)$ then $\hpi$ is a CRS on $\P(\I\restriction_{M^c})$. Moreover, if $\pi$ is monotone then $\hpi$ is monotone. 
% \end{observation}
% \begin{proof}
% Set $\I'=\I\restriction_{M^c}$. Let $E\subset \U$ and let $\x\in b\P(\I')$. Then (1) $\hpi(E)=\pi(E\setminus M)\subset E\setmius M\cap \supp(\x)\subset E\cap \supp(\x)$; (2) $\hpi(E)=\pi(E\setminus M)\in \I$, and since $\pi(E\setminus M)\subset E\setminus M$ it follows that $\hpi(E)\in \I'$ (with probability 1); and (3) for all $e\in \supp(\x)$ $\Pr[e\in \hpi(E)|e\in E]$
% \end{proof}

% The authors of \cite{chekuri2011} also observe the following: 

% \begin{lemma}
% \label{lem:strict_cr}
% Suppose $\Iin$ has a $(b,c)$ CRS. Then it has a ($bc$) CRS, and for any $\epsilon>0$ it has a $\eps$-strict $(b,c)$ (hence an $\eps$ strict $(bc)$) CRS. 
% \end{lemma}

% The $\eps$-strict scheme is attained via random sampling. \ben{Todo: Explain more}. 

\subsection{A general probing strategy}
\label{sec:greedy_prob}
A \emph{permutation} or \emph{ordering} on a subset $E\subset\N$ is a bijection $\sigma:[n]\to E$. To say \emph{traverse $E$ in the order of $\sigma$} should be taken to mean iterate over $E$ in the order $\sigma(1),\sigma(2),\dots,\sigma(|E|)$. The following probing procedure will be repeated often enough throughout the paper to warrant a name. \\

\textbf{$\greedyProb(\sigma,\I_1,\I_2 )$ on $E\subset \U$}. Let $P\gets \emptyset$ and $S\gets\emptyset$. For all $e\in E$ in the order $\sigma$ if $P+e\in\I_1$ and $S+e\in\I_2$, then 
\begin{enumerate}
    \item if $f_S(e)\geq 0$ then probe $e$ and set $P\gets P+e$. If $e$ was active, set $S\gets S+e$. 
    \item if $f_S(e)<0$, then add $e$ to $S$ with probability $p_e$. 
\end{enumerate}
Return $S,P$. \\

%If $\I_1=\I_2=\I$ we may simply write $\greedyProb(\sigma,\I)$. 
We observe that this strategy does in fact produce a valid set of probed elements and a valid solution. 

\begin{observation}
\label{obs:greedy_prob}
Let $S$ and $P$ be the sets returned by \greedyProb. Let $\overline{S}$ be the set of elements queried by \greedyProb which were active. Then $P\cap\act=\overline{S}\subset S$.  Moreover, $S\in\I_2$ hence $\overline{S}\in\I_2$ and $P\in\I_1$ (assuming $\I_1$ and $\I_2)$ are downward closed).  
\end{observation}

\section{Modular SPP}
\label{sec:modular_spp}
The main result of this section is to verify the intuitive result that when the objective function, $f$, is modular, we can obtain the same approximation ratio to the optimal as in SP (i.e., with no prices). Recall that if $f$ is modular, there exist weights $(w_e)_{e\in\U}$ such that $f(E)=\sum_{e\in E}w_e$. Here, we will not draw a distinction between online and offline SPP. This is because the results apply in both settings. We will reduce an instance of SPP to an instance of SP. The result requires only a change of variables, and thus only knowledge of the weights, prices, and activation probabilities in the SPP instance. Thus, if the algorithm for SP is online or applies in any other setting\footnote{For example, stochastic probing with deadlines has been examined \cite{gupta2013,feldman2016online}.}, so too does the corresponding SPP algorithm. 

We begin with an observation that no algorithm worth its salt will query elements whose expected weight is upper bounded by their price.

\begin{lemma}
\label{lem:noDebt}
If $Q$ is a query strategy which queries $e$ where $w_ep_e\leq \price_e$, there exists a query strategy $Q'$ which does not query $e$ such that $\E[Q']\geq \E[Q]$. 
\end{lemma}
\begin{proof}
Given $Q$, define $Q'$ to be the strategy obtained by mimicking $Q$ but refraining from querying any element $e$ such that $w_ep_e\leq\price_e$. It follows that $Q'$ is still a valid strategy because $\Iin$ and $\Iout$ are downward-closed. Moreover, any element queried by $Q$ which is not queried by $Q'$ adds a non-positive expected value to the solution. 
\end{proof}

Henceforth we will apply Lemma \ref{lem:noDebt} and assume that for all elements $e$, $w_ep_e>\price_e$. In this section we will allow prices to be negative; therefore, it is not necessarily the case that $p_e>0$. This section requires that we make comparisons between solutions of SSP and SP, and hence need to introduce the relevant notation. For a probing strategy $Q$, let $\E[Q^\sp(\w,\p)]$ denote the expected value in the SP setting with weights $\w=(w_e)_{e\in\U}$ and probabilities $\p=(p_e)_{e\in\U}$. Similarly, let $\E[Q^\spp(\w,\p,\bpi)]$ denote the expected value of the solution in the SSP setting with weights $(w_e)$, probabilities $(p_e)$ and prices $(\price_e)$. 

Given an instance of SSP, we define new weights and activation probabilities as follows. For all $e$, let $z_e=w_e-\price_e/p_e$ if $p_e>0$ and $-\price_e$ otherwise and let $\hat{p}_e = p_e$ for all $e$ with $p_e>0$, and $\hat{p}_e=1$ for $e$ with $p_e=0$.

\begin{lemma}
\label{lem:sp=spp}
For any querying strategy $Q$ and any weights $\w$, activation probabilities $\p$ and prices $\bpi$,
$\E[Q^\spp(\w,\p,\bpi)]=\E[Q^\sp(\z,\hat{\p})]$. 
\end{lemma}
\begin{proof}
Let $\hat{\act}$ be the random variable denoting the set of active elements according to the probabilities $\hat{\p}$. Note that $p_e>0$ for all $e$, hence $\Pr[e\in Q\cap\hat{\act}]/\hat{p}_e=\Pr[e\in Q]$. Unwinding definitions now gives
\begin{align*}
    \E[Q^\spp&(\w,\p,\bpi)] = \sum_{e\in\N}w_e\Pr[e\in Q\cap \act]-\price_e\Pr[e\in Q]\\
    &=\sum_{e:p_e>0}(w_ep_e-\price_e)\Pr[e\in Q]-\sum_{e:p_e=0}\price_e\Pr[e\in Q]\\
    &=\sum_{e\in\U}z_e\Pr[e\in Q\cap\hat{\act}] =\E[Q^\sp(\z,\hat{\p})].\qedhere
\end{align*}
\end{proof}

We can now prove Theorem \ref{thm:sp=spp}.

\begin{proof}
Let $\alg$ be an algorithm for linear SP which obtains an $\alpha$-approximation. Let $\w,\p,\bpi$ be an instance of linear SPP. Running $\alg$ on weights $\z$ and probabilities $\hat{\p}$ and applying Lemma \ref{lem:sp=spp} gives 
\begin{align*}
    \E[\alg^\spp(\w,\p,\bpi)]=\E[\alg^\sp(\vb{z},\hat{\vb{p}})] 
    &\geq \alpha\E[\opt^\sp(\vb{z},\hat{\vb{p}})] = \alpha \E[\opt^\spp(\vb{w},\vb{p},\bpi)].\qedhere
\end{align*}
\end{proof}

\section{Submodular SPP}
\label{sec:submod_spp}
We now proceed to the more general problem of submodular SPP. We assume in this section that the prices are non-negative. We will employ the common approach of solving a relaxed linear program (Section \ref{sec:fractional_solution}), and then rounding the solution to obtain a probing policy (Section \ref{sec:policy}). First we observe a natural upper bound on the value of the optimal solution, against which we can gauge the quality of approximations. It will be notationally convenient to work with a single polytope instead of both $\P(\Iin)$ and $\P(\Iout)$. Accordingly, we will henceforth let $\P$ refer to the polytope
\[\bigg\{\x\in[0,1]^\U:\x\in\P(\Iout),\;\x\circ\p\in\P(\Iin)\bigg\}.\]

Let $C$ be the (multi)linear extension of $\cost$. Thus $C(\x)=\sum_{e\in\U}\price_ex_e$. A natural relaxation of submodular SPP is the following program: 
\begin{equation}
    \label{eq:submod_relax}
    \max_{\x}\bigg\{ f^+(\x\circ \p)-C(\x): \x\in\P\bigg\}. \tag{LP}
\end{equation}

The relationship between SPP algorithms and LP is given by the following lemma. 

\begin{lemma}
\label{lem:LP-relax}
If $\alg$ be any (stochastic) probing strategy then there exists a point $\x\in \P$ such that  $\E[f(\alg\cap\act)]\leq f^+(\x\circ\p)$ and $\E[\cost(\alg)]=\sum_{e\in\N}\price_ex_e$.
\end{lemma}
\begin{proof}
%\tm{lack of the definition of $S(\alg)$}
Define $\x$ and $\y$ by $x_e=\Pr[e\in\alg]$ and $y_e=\Pr[e\in \alg\cap\act]$. It is immediate that $\x\in\P(\Iin)$ and $\y\in\P(\Iout)$ if $\alg$ is a valid strategy. First, we claim that $\x$ is feasible solution. For this it suffices to show that $y_e=x_ep_e$. Recall that by the constraints of the problem, if an element is queried, then its addition to the current set of probed elements and to the solution must be allowed by the constraints of $\Iout$ and $\Iin$ respectively. 
Thus, $y_e=\Pr[e\in \alg\cap\act]=\Pr[e\in \alg]\Pr[e\in \act]=x_e p_e$, where we've used the fact that whether a particular element is active or not is fixed a priori. Now, notice that $\E[f(\alg\cap\act)]=\E_{S\sim D}[f(S)]$ where $D$ is a particular distribution such that $\Pr_{S\sim D}[e\in S]\leq p_ex_e$. Conversely, $f^+(\x\circ\p)$ is the maximum over all such distributions, i.e., $f^+(\x\circ\p)=\max_D\E_{S\sim D}[f(S)]$. Therefore, $f^+(\x\circ\p)\geq \E[f(\opt\cap\act)]$. Moreover, by the linearity of $\cost$ it's easy to see that $\E[\cost(\alg)]=\sum_{e\in\N}\price_e\Pr[e\in \alg]=\sum_{e\in\N}\price_ex_e$. 
\end{proof}

\subsection{Obtaining a fractional solution}
\label{sec:fractional_solution}
While the problem of maximizing a (non-monotone) submodular function subject to various constraints has been the subject of intense study (e.g., \cite{chekuri2011,feldman2011,lee2009non,feige2011maximizing}), less is known about combinations of submodular functions. In our case, the difficulty in solving \eqref{eq:submod_relax} efficiently arises because the function $f^+(\x\circ\p)-C(\x)$ is not necessarily non-negative. 
Removing the non-negativity condition in a non-monotone submodular maximization problem makes the problem intractable in general, since, as noted in \cite{chekuri2011}, it may take an exponential number of queries to determine whether the optimum is greater than zero. We must therefore take advantage of the special form of our problem; namely the fact that $C$ is linear. 

Recently, Sviridenko et al. gave an approximation algorithm for maximizing the sum of a non-negative, normalized, monotone, submodular function $f$ and a linear function over a matroid constraint \cite[Theorem 3.1]{sviridenko2017}. %Importantly, the approximation factor on the linear term is one. 
%\tm{the above description will be confusing because it uses ``$(1-1/e)$-approximation'' in (slightly) non-standard sense. I think we can provide the precise definition (=below) without the above comment.}
More precisely, given a submodular function $g$, a linear function $\ell$ and a matroid $\I$, with high probability they obtain a set $E$ such that $g(E)+\ell(E)\geq (1-1/e)g(I)+\ell(I)$ minus an arbitrarily small constant term, for any base $I\in\I$. The idea is elegant and straightforward, and involves using the traditional continuous greedy algorithm but over the polytope $\P\cap\{\x:\ell(\x)\geq \lambda\}$ (rather than simply $\P$) where $\lambda$ is a guess for the value of $\ell(\opt)$. Intuitively, this guarantees that the fractional solution $\x^*$ satisfies $L(\x^*)\geq \ell(\opt)$ (where $L$ is the linear extension of $\ell$). Somewhat surprisingly, restricting the polytope in this way does not damage the approximation to $f$. 

Using Measured Continuous Greedy~\cite{feldman2011} instead of continuous greedy and the modification of its proof used in \cite{adamczyk2015} and \cite{adamczyk2018random}, we are able to extend the approach of Sviridenko et al. to 
non-monotone submodular functions and arbitrary constraints. The result is summarized as Theorem \ref{thm:submod+linear}. We say a function $G:[0,1]^\N\to\R$ can be \emph{efficiently estimated with high probability} if for any $\x\in[0,1]^\U$, $G(\x)$ can be determined to within exponentially small error with a polynomial number of queries. For example, if we are given oracle access to a submodular function $f$, then its multilinear extension can be efficiently estimated with high probability (see, e.g., \cite{chekuri2011}).

\begin{theorem}
\label{thm:submod+linear}
Let $g:2^\N\to\R_{\geq0}$ be a normalized submodular function with multilinear extension $G$ which can be efficiently estimated with high probability. Let $\ell:2^\N\to\R$ be a non-decreasing modular function with (multi)linear extension $L$. For any $\eps>0$, $T\geq 0$, and downward-closed system $\I\subset 2^\N$, there exists a polynomial time algorithm which produces a point $\x/T\in \P(\I)$ such that \[G(\x)-L(\x)\geq \alpha(T)g^+(\y)-TL(\y)-O(\eps)R,\]
for any $\y\in\P$ with high probability, where
$R=\max\{\max_{e\in \N}g(e),\max_{e\in\N}|cost(e)|\}$ and $\alpha(T)=1-e^{-T}$ if $g$ is monotone, and $Te^{-T}$ otherwise. 
\end{theorem}

The proof of Theorem \ref{thm:submod+linear} can be found in the Appendix. 

%Due to the similarity of Theorem \ref{thm:submod+linear} to Theorem 3.1 in \cite{sviridenko2017} and of its proof to that of the Measured Continuous Greedy algorithm in \cite{feldman2011}, its proof is moved to the appendix. 

\subsection{Obtaining a policy}
\label{sec:policy}

While Theorem \ref{thm:submod+linear} guarantees the existence of a \emph{fractional} solution which gives a good approximation to the value of the optimal policy, it does not tell us how to query the elements. This section focuses on extracting a probing policy from this fractional solution. We give two rounding techniques, which apply depending on what kind of CRSs to which one has access (i.e., online or offline). Let $\x\in b\P$. \\

% Let $\piin$ be an ordered $(b,c_\i)$ CRS $\Iin$ and $\piout$ be a $(b,c_\o)$ CRS for $\Iout$. We will demonstrate that the algorithm proposed by Gupta and Naragajan \cite{gupta2013} yields a good approximation in the more general setting of submodular objective functions, and that the algorithm of Feldman et al. \cite{feldman2016online} for adversarial SP extends to this setting and to non-monotone submodular functions. The two algorithms are described below. For all three, let $\x\in b\cdot\P$. Furthermore, when running $\greedyProb$ during the algorithms, we will make the following modification: If $S$ is the current solution and $e_i$ is the next element to be queried, we do not query $e_i$ if $f_S(e_i)<0$. \\

% Let $S$ be the final set given by \greedyProb. Let $\overline{S}\subset S$ be the set of elements with non-negative marginals. 

\textbf{Offline-Rounding.} Let $\piin$ and $\piout$ be offline CRSs for $\Iin$ and $\Iout$ respectively. Draw $R\sim \x$ and compute $Q=\piout(R)$. 
%Order $Q$ by $\sigma_\i$ where $\sigma_\i$ is the ordering given by the inner CRS, $\piin$. 
Run $\greedyProb(\sigma,2^\U)$ on $\piin(R)\cap\piout(R)$ where $\sigma$ is any ordering. \\

\textbf{Online-Rounding}. Let $\Iin^{\x}\subset\Iin$ and $\Iout^{\x}\subset \Iout$ be the (random) subsets given by the inner and outer CRSs. If $\piout$ is not online, take $\Iout^{\x}=\Iout$. Draw $R\sim \x$. Run $\greedyProb(\sigma,\Iout^{\x},\Iin^{\x})$ on $R$ if $\piout$ is online, and on $Q$ otherwise where $\sigma$ is any ordering (even adaptively and adversarially chosen). \\

\textbf{Analysis}. 
For either rounding technique, let $P$, $S$, and $\overline{S}$ be as in Observation \ref{obs:greedy_prob}. Given any run of the algorithm, let $\EE$ be defined as those elements which were either probed and active, or else had negative marginals but were nonetheless added to $S$. Note that $\EE$ is distributed as $R(\p)$ and thus as $\act$. However, it might not be the case that $\EE=\act$ because not all elements in $\EE$ were actually probed. Note that $P$, $S$ and $\overline{S}$ are actually functions of $R$ and $\EE$, and it will oftentimes be helpful to write them as such (i.e., $S=S(R,\EE)$). Finally, we define $J(R,\EE)=\piin(R)\cap\piout(R)\cap \EE$, where $\piin$ and $\piout$ were the inner and outer CRSs used in the rounding technique (the given schemes in the case of offline rounding, and the characteristic schemes in the case of online rounding). The following lemma uses properties of the CRSs to obtain the inequality which is crucial to the main result. 

\begin{lemma}
\label{lem:CRbound}
Let $\piin$ and $\piout$ be monotone $(b,c_\i)$ and $(b,c_\o)$ CRSs respectively. Let $\x\in b\P$ and set $R=R(\x)$. Then $\Pr[e\in \piin(R)\cap \piout(R)|e\in R]\geq \max\{c_\o+c_\i-1,c_\o\cdot c_\i\}$. 
Moreover, $\Pr[e\in P]\leq x_e$ where $P$ is the set of elements probed by the algorithm. 
\end{lemma}
\begin{proof}
When conditioning on the choice of the random set $R$, the randomness in $\piin(R)$ and $\piout(R)$ stems purely from the CRSs themselves. Hence, the events $\{e\in \piin(R)|R\}$ and $\{e\in\piout(R)|R\}$ are independent. Additionally, notice that given $R$ and $e\in R$, the quantity $\Pr[e\in \piin(R)]$ is monotonically decreasing in $R$. Combining these two facts and 
conditining on $e\in R$ gives
\begin{align}
    &\Pr_{R,\piin,\piout}[e\in \piin(R)\cap \piout(R)]\\
    &=\E_R\big[\Pr_{\piin,\piout}[e\in \piin(R)\cap \piout(R)|R]\big] \notag\\
    &=\E_R\big[\Pr_{\piin}[e\in \piin(R)|R]\Pr_{\piout}[e\in\piout(R)|R]\big]\notag\\
    &\geq \E_R[\Pr_{\piin}[e\in\piin(R)|R]]\cdot\E_R[\Pr_{\piout}[e\in\piout(R)|R]]\notag\\
    &=\Pr_{R,\piin}[e\in\piin(R)]\Pr_{R,\piout}[e\in\piout(R)], \label{eq:lem_crbound1}
\end{align}
using the FKG inequality. 
This implies that $\Pr[e\in\piin(R)\cap\piout(R)|e\in R]\geq c_\o\cdot c_\i$. To obtain the other bound, we compute 
\begin{align*}
    &\Pr[e\in\piin(R)\cap\piout(R)] \\
    &= \Pr[e\in \piout(R)] - \Pr[e\in \piin(R)^c\cap \piout(R)] \\
    &\geq c_\o-\Pr[e\in\piin(R)^c]
    = c_\o - (1-\Pr[e\in\piin(R)]) \\
    &\geq c_\o-(1-c_\i).
    % &=c_\o- \bigg(1- \frac{\Pr[e\in\piin(R)\cap e\in\piout(R)]}{\Pr[e\in\piout(R)]}\bigg) \\
    % &\geq c_\o - (1-\Pr[e\in\piin(R)])\geq c_\o+c_\i-1,
\end{align*}
Rearranging gives the desired result. The final statement follows from noticing that $P\subset R$ in either rounding scheme, and $\Pr[e\in R]=x_e$. 
\end{proof}

Instead of reasoning directly about the set of elements in the solution, we will reason about the set $J(R,\EE)$. The following two lemmas prove useful properties about this set. 

\begin{lemma}
\label{lem:joint_monotone}
Given $R_1\subset R_2$ and $\EE_1\subset\EE_2$ with $e\in R_1\cap\EE_1$, $\Pr[e\in J(R_1,\EE_1)]\geq \Pr[e\in J(R_2,\EE_2)]$.  
\end{lemma}
\begin{proof}
The randomness comes only from the CRSs, since $R_1,R_2,\EE_1$ and $\EE_2$ are given. Therefore, $\Pr_{\piin,\piout}[e\in\piin(R_1)\cap\piout(R_1)\cap \EE_1]=\Pr_{\piin}[e\in\piin(R_1)]\Pr_{\piout}[e\in\piout(R_1)]\geq\Pr_{\piin}[e\in\piin(R_2)]\Pr_{\piout}[e\in\piout(R_2)]=\Pr_{\piin,\piout}[e\in\piin(R_2)\cap\piout(R_2)\cap \EE_2]$.
\end{proof}

\begin{lemma}
\label{lem:joint_solution}
For any $R$ and $\EE$, 
$J(R,\EE)\subset S(R,\EE)$. 
\end{lemma}
\begin{proof}
Let $e\in\joint(R,\EE)$. We first observe that regardless of the rounding technique, if $e$ advances to having its marginal considered in \greedyProb then, by definition of $\EE$, it will be added to $S=S(R,\EE)$. Therefore, it remains only to show that $e$ meets the condition $P_e+e\in\I_1$ and $S_e+e\in\I_2$ in \greedyProb, where $P_e\subset P$ and $S_e\subset S$ are the respective intermediary solutions immediately before $e$ is considered by the algorithm.
In the case of offline rounding, this is immediate since \greedyProb is run on $\piout(R)\cap\piin(R)$ which is in $\Iout=\I_1$ and $\Iin=\I_2$ (since $\piout(R)\subset \Iout$ and $\piout(R)\subset \Iin$ w.p. 1).
Now consider online rounding with two online schemes. Here, recall that $\piin(R)=\{e\in R:I+e\in \Iin,\forall I\subset R,I\in\Iin\}$. Taking $I=S_e$ we see that $S_e+e\in\Iin$. Similarly, $P_e+e\in\Iout$. The argument for online rounding when $\piout$ is offline is similar. 
\end{proof}

This final technical lemma derives a lower bound on the marginal of our solution with respect to that of the optimal's. The main approximation guarantee will then result from decomposing the objective function into the sum of its marginals and applying the following lemma.  

\begin{lemma}
\label{lem:marginal_approx}
Let $\x\in b\P$ and let $\{e_1,e_2,\dots,e_n\}$ be any ordering on $\U$. 
For every $e_i\in\U$, $\E[f_{\overline{S}^{i-1}} (e_i)]\geq \gamma\E[f_{R^{i-1}(\x)\cap \act}(e_i)]$ where $\gamma=\max\{c_\o+c_\i-1,c_\o c_\i\}$.
\end{lemma}
\begin{proof}
Fix $i\in[n]$ and let $e=e_i$. Set $R=R(\x)$. We have: 
\begin{align*}
    \E_{\EE,\pi}[f_{\overline{S}^{i-1}}(e)] 
    &= \Pr[e\in R\cap \EE]\E_{\EE,\pi}[\ind(e\in \overline{S})f_{\overline{S}^{i-1}}(e)|e\in R\cap \EE] \\
    &= \Pr[e\in R\cap 
    \EE]\E_{\EE,\pi}[\ind(e\in \overline{S})\max\{0,f_{\overline{S}^{i-1}}(e)\}|e\in R\cap \EE] \\
    &= \Pr[e\in R\cap \EE]\E_{\EE,\pi}[\ind(e\in S)\max\{0,f_{\overline{S}^{i-1}}(e)\}|e\in R\cap \EE] \\
    &\geq \Pr[e\in R\cap \EE]\E_{\EE,\pi}[\ind(e\in S)\max\{0,f_{R^{i-1}\cap \EE}(e)\}|e\in R\cap \EE],
\end{align*}
where the third equality follows from the fact that if $e\in S\setminus\overline{S}$ then $\max\{0,f_{\overline{S}^{i-1}}(e)\}=0$, and the final inequality follows from submodularity.
Let $\phi(R,\EE)=\max\{0,f_{R^{i-1}\cap \EE}(e)\}$. Now, condition on $e\in R\cap \EE$ and write
\begin{align*}
    \E_{R,\EE,\pi}[\ind(e\in \overline{S}(R,\EE))\cdot\phi(R,\EE)] 
    &= \E_{R,\EE}\big[\E_{\pi}[\ind(e\in \overline{S}(R,\EE))\cdot\phi(R,\EE)|R,\EE]\big] \\
    &= \E_{R,\EE}\big[\E_\pi[\ind(e\in \overline{S}(R,\EE))|R,\EE]
    \cdot\phi(R,\EE)\big] \\ 
    &\geq  \E_{R,\EE}\big[\E_\pi[\ind(e\in \joint(R,\EE)|R,\EE]\cdot \phi(R,\EE)\big],
\end{align*}
where the final inequality follows from Lemma \ref{lem:joint_solution}. Both $\phi(R,A)$ and  $\E[\ind(e\in J(R,\EE)|R\cap\EE]$ are decreasing functions of $R,\EE$ (the former from submodularity and the latter from Lemma \ref{lem:joint_monotone}). Therefore, by the FKG inequality, the above is at least 
\begin{align*}
    &\E_{R,\EE}[\E_\pi[\ind(e\in \joint(R,\EE))|R,\EE]|e\in R\cap \EE]\cdot \E_{R,\EE}[\phi(R,\EE)|e\in R\cap\EE] \\
    &= \E_{R,\EE,\pi}[\ind(e\in \joint(R,\EE))|e\in R\cap \EE]\cdot \E_{R,\EE}[\phi(R,\EE)|e\in R\cap \EE] \\
    &= \E_{R,\pi}[\ind(e\in \piin(R)\cap\piout(R))|e\in R]\cdot \E_{R,\EE}[\phi(R,\EE)|e\in R\cap \EE] \\
    &\geq \gamma \E_{R,\EE}[\phi(R,\EE)|e\in R\cap\EE],
\end{align*}
by Lemma \ref{lem:CRbound},
where the final equality uses the fact that the events $\{e\in\EE\}$ and $\{e\in\piin(R)\cap\piout(R)\}$ are independent. Combining everything and keeping in mind that $R\cap\EE$ is distributed as $R\cap \act$, we obtain
\begin{align*}
    \E[f_{S^{i-1}}(e)]\geq
    &=\gamma \Pr[e\in R\cap \EE]\E[\phi(R,\EE)|e\in R\cap \EE]\geq \gamma \E[f_{R^{i-1}(\x)\cap \act}(e)],
\end{align*}
as desired. 
\end{proof}

\subsection{Approximation Guarantees}
\label{sec:approximations}
Given the rounding policies presented in the previous section, we are now ready to prove Theorem \ref{thm:spp}, and explore the guarantees given in the case where $\Iin$ and $\Iout$ are the intersection of matroids. 

\subsubsection{Proof of Theorem \ref{thm:spp}}

\begin{proof}
Apply Theorem \ref{thm:submod+linear} with $G(\x)=F(\x\circ \p)$ and $L(\x)=C(\x)$ and run CRS-Rounding on the resulting point $\x\in b\cdot\P$. Then, applying Lemmas \ref{lem:CRbound} and \ref{lem:marginal_approx} gives 
\begin{align*}
    \E[\alg]&=\E[f(\alg\cap \act)]-\E[\cost(\alg)]\\
    &=f(\emptyset)+\sum_{i=1}^n \E[f_{\alg^{i-1}\cap \act}(e_i)] -\sum_{e\in\U}\price_e\Pr[e\in \alg]\\
    &\geq  f(\emptyset)+\gamma\sum_{i=1}^n  \E[f_{R^{i-1}\cap \act}(e_i)] - \sum_{e\in\U}\price_ex_e \\
    &= \gamma \E[f(R(\x)\cap \act)]-C(\x).
\end{align*}
Applying Lemma \ref{lem:LP-relax}, let $\x^*$ satisfy $\E[f(\opt\cap\act)]\leq f^+(\x^*\circ\p)$ and $\E[\cost(\opt)]=\sum_{e\in\N}\price_ex_e^*$, where $\opt$ is an optimal probing strategy.
Noticing that $\E[f(R(\x)=F(\x\circ\p)$,  we have 
\begin{align*}
    \gamma\E[f(R(\x)\cap \act)]-C(\x) 
    &\geq \gamma\alpha(b)f^+(\x^*\circ\p)-b L(\x^*)-\gamma O(\eps) R, 
\end{align*}
with high probability by Theorem \ref{thm:submod+linear}. Lemma \ref{lem:LP-relax} implies that $\alg$ is a $(\gamma \alpha(b),b)$-approximation. 

\end{proof}

\subsubsection{Intersection of matroids for SP}
Here we focus on traditional Stochastic Probing when $\Iin$ and $\Iout$ are the intersection of $k$ and $\ell$ matroids, respectively, we have the existence of $(b,1-b)$ offline CRSs and online CRSs, for any $b\in[0,1)$ \cite{feldman2016online,gupta2013}. 
For any fixed value of $b$, Theorem \ref{thm:spp} gives a 
$\max\{(1-b)^{k+\ell},(1-b)^k+(1-b)^\ell-1\}(1-e^{-b})$
approximation. 
By an easy induction, we see that for all $k,\ell\in\mathbb{N}$, $(1-b)^{k+\ell}\geq (1-b)^k+(1-b)^\ell-1$. Therefore, in the case of matroids, it's always more optimal to use the online CRSs than the offline CRSs and we obtain a $\max_{b\in[0,1)}\{(1-b)^{k+\ell}(1-e^{-b})\}$ approximation. Next, we determine the optimal value of $b$ when the objective function is monotone.

\begin{lemma}
\label{lem:kl_matroid_nonm}
For any $k,\ell\in\mathbb{N}$, 
\[(k+\ell+1)-W((k+\ell)e^{k+\ell+1})=\argmax_{b\in(0,1]}\{(1-b)^{k+\ell}(1-e^{-b})\}.\]
\end{lemma}
\begin{proof}
Let $z=k+\ell$ and set $\xi(b)=(1-b)^z(1-e^{-b})$. Solving $\xi'(b)=0$ yields $(1-b)^ze^{-b}=z(1-b)^{z-1}(1-e^{-b})$, i.e., $(1-b+z)e^{-b}=z$ which is solved by $b=z+1-W(ze^{z+1})\equiv b^*$. The only other critical points of $\xi$ in the region $0\leq b\leq 1$ are $b=0,1$, which yield $\xi(b)=0$. Thus, $b^*$ is the unique maximizer of $\xi(b)$. 
\end{proof}

\begin{corollary}
If $\Iin$ and $\Iout$ are the intersection of $k$ and $\ell$ matroids then there exists a $(W(ze^{z+1})-z)^{z+1}\big/W(ze^{z+1})$-approximation to non-monotone, adversarial SP where $z=k+\ell$. 
\end{corollary}
\begin{proof}
Evaluate the approximation ratio at $b=k+\ell-W((k+\ell)e^{k+\ell+1})$. 
\end{proof}

We now observe that the above approximation is state-of-the-art compared to known results for online, monotone SP. Recall that Feldman et al. give a $c_\i c_\o\cdot(1-e^{-b})$-approximation in the adversarial setting, and Adamczyk and W\l{}odarcyzk~\cite{adamczyk2018random} give a $\frac{1}{(k+\ell+1)e}$-approximation. The former is less than $\max\{c_\i c_\o,c_\i+c_\o-1\}(1-e^{-b})$. For the latter we use the following lemma. 

\begin{lemma}
\label{lem:monotone_sp_better}
$\max_{b \in [0,1]} (1 - e^{-b})(1 - b)^{k + \ell} \ge 1/(k + \ell + 1) e$.
\end{lemma}
\begin{proof}
We evaluate the left-hand side at $b = t = 1/(k+\ell+1)$. 
Let $\phi(t) = (1 - t)^{(1/t) - 1}$.
Then, our goal is to show that $(1 - e^{-t}) \phi(t) \ge t/e$ for $t = 1/3, 1/4, \dots$.
This is easily verified at $t = 1/3$ by direct computation.
For $t \le 1/4$, we use
\begin{align}
    1 - e^{-t} \ge t - \frac{t^2}{2}.
\end{align}
Also, we can show that $\phi(0) = 1/e$, $\phi'(0) = 1/2e$, and $\phi''(t) \ge 7/12e$.
Therefore, by Taylor's theorem, 
\begin{align*}
    \phi(t) \ge \frac{1}{e} + \frac{t}{2e} + \frac{7 t^2}{24e}.
\end{align*}
Thus,
\begin{align*}
    (1 - e^{-t}) \phi(t) \ge \frac{t}{e} \left(1 + \frac{t^2}{24} - \frac{7t^3}{48} \right) \ge \frac{t}{e}. \qedhere
\end{align*}
\end{proof}

Finally, we give the optimal value of $b$ when the objective function is non-monotone. 

\begin{lemma}
\label{lem:kl_matroid}
For any $k,\ell\in\mathbb{N}$, 
\[\frac{1}{2}(k+\ell+2-\sqrt{(k+\ell)(k+\ell+4)})=\argmax_{b\in(0,1]}\{(1-b)^{k+\ell}\cdot be^{-b})\}.\]
\end{lemma}
\begin{proof}
Let $z=k+\ell$ and set $\xi(b)=(1-b)^zbe^{-b}$. Then $\xi'(b)=0$ iff $e^{-b}(1-b)^{z-1}(1-(2+z)b+b^2)=0$. We require that $b>0$, implying that $b=\frac{1}{2}(z+2-\sqrt{z(z+4)})\equiv b^*$ (it could not have been the other root since $b\leq 1$). It is easy to verify that $0< b\leq 1$. Now, \[\xi''(b) = -\xi'(b) - (z-1)(1-b)^{-1}\xi'(b) + e^{-b}(1-b)^{z-1}(2b-(2+z)),\]
and so $\xi'(b^*)=e^{-b^*}(1-b^*)^{z-1}(2b-(2+z))<0$ since $e^{-b^*},(1-b^*)^{z-1}>0$ and $2b^*-(2+z)<0$ because $2b^*\leq 2<2+z$.  
\end{proof}

\paragraph{Acknowledgements}
We would like to thank Moran Feldman for pointing out an error in an earlier version of the paper.

\bibliographystyle{plain}
\bibliography{ref.bib}

\appendix

\section{Proof of Theorem \ref{thm:submod+linear}}

Let $g$ be a normalized submodular function and $c$ a non-decreasing modular function. Let $G$ and $C$ be their respective multilinear extensions, and let $T\geq 0$ be given. 
Throughout the proof we will assume that we have an oracle to evaluate $G$. Without such an oracle, we may evaluate $G$ with high probability with polynomially many queries. This is standard practice (e.g., ~\cite{chekuri2011}).

Recall that $R=\max\{\max_{e\in\N} g_\emptyset(e), \max_{e\in\N}|\cost(e)|\}$ and note that $g(OPT)$ and $|\cost(OPT)|$ are both upper bounded by $nR$.

Let $\y^*\in\argmax_{\y} \{\alpha(T)g^+(\y)-T C(\y)\}$. Recall that our goal is to find a point $\x\in T\cdot\P$ such that $G(\x)-C(\x)\geq \alpha(T) g^+(\y^*)-T C(\y^*)$. Our first goal is to estimate the value of $C(\y^*)$. As in \cite{sviridenko2017}, we do this by sampling $O(\eps^{-1} n\log n)$ points from the interval $[-nR,nR]$. For each point $\theta$, we will essentially run Measured Continuous Greedy over the polytope $\P\cap \{\x:\wC(\x)\leq \theta \}$. The algorithm is described formally below.\\

\textbf{Modified Measured Continuous Greedy.}
Assume that $1/\eps\in\mathbb{N}$; otherwise decrease $\eps$ sufficiently. Similarly to \cite{sviridenko2017}, fill $[0,TR]$ with $O(\eps^{-1})$ points of the form $i\eps TR$ for $i\in\{0,1,\dots,\eps^{-1}\}$, and $[TR,TnR]$ with $O(\eps^{-1}\log n)$ points of the form $(1+\eps/n)^i\log (Tn)$ for $i\in\{0,1,\dots,\lceil\log_{1+\eps/n}n\rceil\}$. For each point $\theta$, perform the following. 

Set $\x^0_\theta=\ind_{\emptyset}$ and assume $\delta$ is sufficiently small. Let $\P(\theta)=\P\cap \{\x:\wC(\x)\leq \theta\}$. For $i=k\delta$, $k=0,1,\dots,T/\delta-1$ let $\v^i=\argmax_{\v}\{\v\cdot(\grad G(\x_\theta^i)\circ(1-\x_\theta^i)):\v\in\P(\theta)\}$. Define $\x^{i+\delta}_\theta$ by $x^{i+\delta}_e=x^i_{\theta,e}+\delta v_e^i (1-x_{\theta,e}^i)$. As the final solution we return $\argmax_\theta G(\x_\theta^T)-C(\x_\theta^T)$. 

We emphasize the similarity of this procedure to Measured Continuous Greedy \cite{feldman2011}. The only differences are the estimation of $C(\y^*)$ and restriction of the given polytope $\P$ to $\P(\theta)$. \\

\textbf{Analysis.} 
For some point $\theta$ we have $\theta\leq C(\y^*)\leq \theta+\eps R$; see \cite{sviridenko2017} for more details. We will perform the rest of the analysis for this value of $\theta$. For notational simplicity let $\x=\x^T_\theta$. We want to demonstrate that
\[G(\x)-C(\x)\geq \alpha(T) g^+(\y^*)-T C(\y^*)-O(\eps)R\] with high probability, and that $\x/T\in \P$. The argument of the latter fact does not change from the analysis of Measured Continuous Greedy in \cite{feldman2011}, thus we proceed to prove the former.

We begin by upper bounding $C$. Noting that $\x=\sum_{i=0}^{T/\delta-1}\delta \v^{i\delta}\circ(1-\x^{i\delta})$ and that $\wC$ is linear and non-decreasing, we have 
\begin{equation*}
    \wC(\x)=\delta\sum_{i=0}^{T/\delta-1}\wC(\vb{v}^{i\delta}\circ(1-\x^{i\delta}))\leq \delta\sum_{i=0}^{T/\delta-1}\wC(\vb{v}^{i\delta})\leq T\theta,
\end{equation*}
since $\wC(\vb{v}^{i\delta})\leq \theta$ for each $i$ by construction. Therefore, 
$C(\x)= T\theta\leq  T C(\y^*)$.

Now we proceed to lower bounding $G$. Here we follow the analysis of Measured Continuous Greedy as found in \cite{feldman2013} and the modification of its proof as in \cite{adamczyk2015} and ~\cite{adamczyk2018random}. We begin by assuming that $g$ is non-monotone. We need the following lemmas from Feldman et al. 

\begin{lemma}
[\cite{feldman2011}]
If $|x_e'-x_e|\leq \delta$ for all $\e\in\U$, then 
\begin{equation}
    \label{eq:feldman2011_lemma}
    G(\x')-G(\x)\geq \sum_{e\in\U}(x_e'-x_e)\partial_e G(\x)-O(n^3\delta^2)\max_e f(e). 
\end{equation}
\end{lemma}

\begin{lemma}[\cite{feldman2011}]
For all $i$, and $e\in\U$, $x_e^i\leq 1-(1-\delta)^i/\delta\leq 1-e^{-i}+O(\delta)$. Moreover, if $x_e\leq a$ for every $ e\in \U$ then for all $E\subset\U$, $G(\x\vee\ind_E)\ge (1-a)g(E)$. 
\end{lemma}

Let $H_e(\x^i)=G(\x_i\vee\ind_e)-G(\x_i)$ and define $H_E(\x_i)$ similarly. Applying \eqref{eq:feldman2011_lemma} to $\x^{i+\delta}-\x^i$ and noting that $\partial_e G(\x)=\frac{G(\x\vee\ind_e)-G(\x)}{1-x_e}$ yields 
\begin{align}
\label{eq:mmcg1}
    G(\x^{i+\delta})-G(\x^i)\geq \delta \sum_{e\in\N} v_e^iH_e(\x^i)-O(n^3\delta^2)\max_e g(e) \geq \delta\sum_{e\in\U}y_e^*H_e(\x^i)-O(n^3\delta^2)\max_e g(e), 
\end{align}
where the second inequality uses the definition of $\v^i$. Let $(p_E)_{E\subset \U}$ be the maximizing argument of $g^+(\y^*)$, i.e., $g^+(\y^*)=\sum_{E\subset\U}p_E g(E)$.  Then,
\begin{align}
\label{eq:mmcg2}
    \sum_{e\in\N}y_e^* H_e(\x_i)&=\sum_{E\subset\U}p_E\sum_{e\in E}H_e(\x_i)\geq \sum_{E\subset\N}p_E H_E(\x_i) \notag \\ 
    &\geq \bigg(\sum_{E\subset\N} p_E g(E)(e^{-i}-O(\delta))\bigg)-G(\x^i) =(e^{-i}-O(\delta))g^+(\y^*) - G(\x^i), 
\end{align}
where the first inequality is due to monotonicity.  
%\tm{first term will be $\sum_{e \in \U}$} \tm{there is no definition of $H_E$} \tm{second inequality is by monotonicity}
% Combining \eqref{eq:mmcg1} and \eqref{eq:mmcg2} gives
% \begin{equation}
% \label{eq:mmcg3}
% G(\x^{i+\delta})-G(\x^i)\geq \delta(e^{-i}g^+(\y^*)-G(\x^i))-O(n^3\delta ^2)\max_e g(e).
% \end{equation}
Now, notice that $\max_e g(e)\leq \max_{E\subset\N}g(E)= \max_{E\subset\N}g^+(\ind_E)\leq g^+(\vb{z}^*)$ where $\z^*\in\argmax_{\vb{z}}\{g^+(\vb{z}):\vb{z}\in\P\}$ (because $\ind_E\in\P$). By definition of $\y^*$ we have $\alpha(T)g^+(\y^*)-T\beta C(\y^*)\geq \alpha(T)g^+(\vb{z}^*)-T\beta C(\vb{z}^*)$. Moreover, $C(\vb{z}^*)=\sum_{e\in\N}\price_ez_e^*\leq \price$ where $\price=\sum_{e\in\N}\price_e$. Hence,
%\tm{$z^*$ bold after "where"}
\begin{equation*}
    \max_e g(e)\leq g^+(\vb{z}^*)\leq g^+(\y^*)  -T\beta g^+(\y^*) + T\beta \price \leq g^+(\y^*) + T\beta \price 
    %-\theta/\alpha(T) + \frac{T\price}{\alpha(T)\beta}.
\end{equation*}
%\tm{is there some factor of $\alpha(T)$ in the constant term? (do we obtain this by dividing $\alpha(T)$ both-hand sides of $\alpha(T)g^+(\y^*)-T\wC(\y^*)\geq \alpha(T)g^+(\vb{z}^*)-T\wC(\vb{z}^*)$?)}
Setting $\kappa=T\beta\price$, and combining the above with \eqref{eq:mmcg1} and \eqref{eq:mmcg2} gives 
\begin{equation}
\label{eq:mmcg3}
G(\x^{i+\delta})-G(\x^i)\geq \delta(e^{-i}g^+(\y^*)-G(\x^i))-O(n^3\delta ^2)g^+(\y^*)-O(n^3\delta^2)\kappa.
\end{equation}

As noted in \cite{adamczyk2015}, the analysis of \cite{feldman2013} relies on demonstrating that for all $i$
\begin{equation}
    \label{eq:measured_greedy_bound}
    G(\x^{i+\delta})-G(\x^i)\geq \delta[(e^{-i}\cdot g(OPT)-G(\x^i)]-O(n^3\delta^2)g(OPT).
\end{equation}
The analysis after this point does not use any properties of $g(OPT)$, and can be replaced by any constant. In our case, we will replace it with $g^+(\y^*)$. Additionally (and crucially in our case) the analysis uses no properties of the vector $v_e(t)$, meaning that it is not affected by working over the polytope $\P(\theta)$. Therefore, continuing with their analysis, but tacking on the term $O(n^3\delta^2)\kappa$ onto the end, eventually yields that
\[G(\x^T)\geq [Te^{-T}-O(n^3\delta^2)]g^+(\y^*)-O(n^3\delta^2)\kappa=(Te^{-T}-o(1))g^+(\y^*)-o(1),\]
if $\delta$ is chosen sufficiently small. Combining this inequality with \eqref{eq:measured_greedy_bound} gives the desired result, and completes the proof if $g$ is non-monotone. If $g$ is monotone, then we may strengthen the bound with precisely the same techniques used in \cite{feldman2013}.

\end{document}